\documentclass[11pt]{article}
\usepackage{amsmath}
\usepackage{amssymb}
\usepackage{amsthm}
\usepackage{graphicx}
\usepackage{mathtools}
\usepackage{hyperref}
\usepackage{enumitem}

\usepackage{newpxtext} 
\usepackage[varg,bigdelims]{newpxmath}
\usepackage[margin=1in]{geometry}
\usepackage{todonotes}

\renewcommand{\le}{\leqslant}

\renewcommand{\ge}{\geqslant}

\renewcommand{\Pr}{\text{Pr}}
\newcommand{\N}{\mathbb{N}}
\newcommand{\F}{\mathbb{F}}
\newcommand{\vrk}[1]{\text{vrk}(#1)}
\newcommand{\rk}[2]{\text{rk}_{{#1}}({#2})}

\newcommand{\tensor}{\otimes}

\newcommand{\rank}{\text{rank}}

\newcommand{\abs}[1]{\left|#1\right|}

\renewcommand{\rank}{\text{rank}}

\renewcommand{\Pr}[1]{\textbf{Pr}\left[#1\right]}

\newcommand\restr[2]{{
  \left.\kern-\nulldelimiterspace 
  #1 
  \vphantom{|} 
  \right|_{#2} 
  }}

\newtheorem{theorem}{Theorem}
\newtheorem{prop}[theorem]{Proposition}
\newtheorem{lemma}[theorem]{Lemma}
\newtheorem{corollary}[theorem]{Corollary}

\newtheorem{defn}{Definition}
\newtheorem{claim}{Claim}
\newenvironment{claimproof}[1]{\par\noindent\underline{Proof:}\space#1}{\hfill $\blacksquare$\medskip}

\begin{document}
    
\title{Visible Rank and Codes with Locality\thanks{An earlier version this work was presented at 2021 RANDOM/APPROX~\cite{AG-random21}. The current version includes Theorem~\ref{thm:tensor-gaps}, which is a solution to Question 2 that was asked in~\cite{AG-random21}. This work was done when both authors were affiliated with Carnegie Mellon University.}\thanks{Research supported in part by NSF grants CCF-1814603 and CCF-1908125 and a Simons Investigator award.}}
\author{Omar Alrabiah\thanks{Department of Electrical Engineering and Computer Science, University of California, Berkeley, CA 94704. Email: {\tt oalrabiah@berkeley.edu}} 
\and Venkatesan Guruswami\thanks{Department of Electrical Engineering and Computer Science, University of California, Berkeley, CA 94704. Email: {\tt venkatg@berkeley.edu}.}
}
    
\date{}
\maketitle
\thispagestyle{empty}

\begin{abstract}
We propose a framework to study the effect of local recovery requirements of codeword symbols on the dimension of linear codes, based on a combinatorial proxy that we call \emph{visible rank}. The locality constraints of a linear code are stipulated by a matrix $H$ of $\star$'s and $0$'s (which we call a "stencil"), whose rows correspond to the local parity checks (with the $\star$'s indicating the support of the check). The visible rank of $H$ is the largest $r$ for which there is a $r \times r$ submatrix in $H$ with a unique generalized diagonal of $\star$'s. The visible rank yields a field-independent combinatorial lower bound on the rank of $H$ and thus the co-dimension of the code. 

\smallskip
We point out connections of the visible rank to other notions in the literature such as unique restricted graph matchings, matroids, spanoids, and min-rank. In particular, we prove a rank-nullity type theorem relating visible rank to the rank of an associated construct called \emph{symmetric spanoid}, which was introduced by Dvir, Gopi, Gu, and Wigderson~\cite{DGGW20}. Using this connection and a construction of appropriate stencils, we answer a question posed in \cite{DGGW20} and demonstrate that symmetric spanoid rank cannot improve the currently best known $\widetilde{O}(n^{(q-2)/(q-1)})$ upper bound on the dimension of $q$-query locally correctable codes (LCCs) of length $n$. This also pins down the efficacy of visible rank as a proxy for the dimension of LCCs. 

\smallskip
We also study the $t$-Disjoint Repair Group Property ($t$-DRGP) of codes where each codeword symbol must belong to $t$ disjoint check equations. It is known that  linear codes with $2$-DRGP must have co-dimension $\Omega(\sqrt{n})$ (which is matched by a simple product code construction). We show that there are stencils corresponding to $2$-DRGP with visible rank as small as $O(\log n)$. However, we show the second tensor of any $2$-DRGP stencil has visible rank $\Omega(n)$, thus recovering the $\Omega(\sqrt{n})$ lower bound for $2$-DRGP. For $q$-LCC, however, the $k$'th tensor power for $k\le n^{o(1)}$ is unable to improve the $\widetilde{O}(n^{(q-2)/(q-1)})$ upper bound on the dimension of $q$-LCCs by a polynomial factor. Inspired by this and as a notion of intrinsic interest, we define the notion of \emph{visible capacity} of a stencil as the limiting visible rank of high tensor powers, analogous to Shannon capacity, and pose the question whether there can be large gaps between visible capacity and algebraic rank.

\end{abstract}

\newpage

\section{Introduction}

The notion of \emph{locality} in error-correcting codes refers to the concept of recovering codeword symbols as a function of a small number of other codeword symbols. Local decoding requirements of various kinds have received a lot of attention in coding theory, due to both their theoretical and practical interest. For instance, $q$-query locally correctable codes (LCCs) aim to recover any codeword symbol as a function of $q$ other codeword symbols in a manner robust to a constant fraction of errors. On the other hand, locally recoverable codes (LRCs), in their simplest incarnation, require each codeword symbol to be a function of some $\ell$ other codeword symbols, allowing local recovery from any single erasure.\footnote{There is also a distance requirement on LRCs to provide more global error/erasure resilience.}

LCCs have been extensively studied in theoretical computer science, and have connections beyond coding theory to topics such as probabilistically checkable proofs and private information retrieval. We refer the reader to~\cite{Yek12} and the introduction of~\cite{Gop18} for excellent surveys on LCCs and their connections. LRCs were motivated by the need to balance global fault tolerance with extremely efficient repair of a small number of failed storage nodes in modern large-scale distributed storage systems~\cite{GHSY12}. They have led to intriguing new theoretical questions, and have also had significant practical impact  with adoption in large scale systems such as Microsoft Azur~\cite{HSX+12} and Hadoop~\cite{SAP+13}.

Let us define the above notions formally, in a convenient form that sets up this work. We will restrict attention to linear codes in this work, i.e., subspaces $C$ of $\F_q^n$ for some finite field $\F_q$. In this case, the $i$'th symbol $c_i$ of every codeword $c=(c_1,\dots,c_n) \in C$ can be recovered as a function of the symbols $c_j$, for indices $j$ in a (minimal) subset $R_i \subset [n] \setminus \{i\}$, iff $c_i$ and $\{c_j \mid j \in R_i\}$ satisfy a linear check equation, or in other words, there is a dual codeword whose support equals $\{i\} \cup R_i$. The set $R_i$ is called a \emph{repair group} for the $i$'th codeword symbol (other terminology used in the literature includes regenerating sets and recovery sets).

The $q$-LCC property, for a fixed number of queries $q$ and growing $n$, corresponds to having $\Omega(n)$ disjoint groups of size $\le q$ for each position $i \in [n]$, or equivalently $\Omega(n)$ dual codewords of Hamming weight at most $(q+1)$ whose support includes $i$ and are otherwise disjoint. The $\ell$-LRC corresponds to having a dual codeword of Hamming weight at most $(\ell+1)$ whose support includes $i$, for each $i \in [n]$. A property that interpolates between these extremes of a single repair group and $\Omega(n)$ disjoint repair groups is the Disjoint Repair Group Property ($t$-DRGP) where we require $t$ disjoint repair groups for each position $i \in [n]$ (equivalently $t$ dual codewords whose support includes $i$ but are otherwise disjoint).

There is an exponentially large gap between upper and lower bounds on the trade-off between code dimension and code length for $q$-LCCs. The best known code constructions have dimension only $O((\log n)^{q-1})$ (achieved by generalized Reed-Muller codes or certain lifted codes~\cite{GKS13}), whereas the best known upper bound on the dimension of $q$-LCCs is much larger and equals $\widetilde{O}(n^{(q-2)/(q-1)})$~\cite{KT00,Woodruff12,IS20}\footnote{The $\widetilde{O}(\cdot)$ and $\widetilde{\Omega}(\cdot)$ are used to suppress factors poly-logarithmic in $n$.}. Narrowing this huge gap has remained open for over two decades. 

In contrast, the best possible dimension of a $\ell$-LRC is easily determined to be $\lfloor \tfrac{\ell n}{\ell+1}\rfloor$.\footnote{In this case, a more interesting trade-off is a Singleton-type bound that also factors in the distance of the code~\cite{GHSY12}.} However, for $t$-DRGP, there are again some intriguing mysteries. For $2$-DRGP, we have tight bounds---the minimum possible redundancy (co-dimension) equals $\sqrt{2n} \pm \Theta(1)$. The lower bound is established via very elegant proofs based on the polynomial method~\cite{Woo16} or rank arguments~\cite{RV16}. However, for fixed $t > 2$, we do not know better lower bounds, and the best known constructions have co-dimension $\approx t \sqrt{n}$~\cite{FVY15}. There are better constructions known for some values of $t=n^{\Theta(1)}$~\cite{FGW17,LiW21}. A lower bound on the co-dimension of $c(t) \sqrt{n}$ for some function $c(t)$ that grows with $t$ seems likely, but has been elusive despite various attempts, and so far for any fixed $t$, the bound for $t=2$ is the best known.

This work was motivated in part by these major gaps in our knowledge concerning $q$-LCCs and $t$-DRGPs. Our investigation follows a new perspective based on \emph{visible rank} (to be defined soon), which is a combinatorial proxy for (linear-algebraic) rank that we believe is of broader interest. This is similar in spirit to a thought-provoking recent work~\cite{DGGW20} that introduced a combinatorial abstraction of spanning structures called \emph{spanoids}\footnote{We defer a precise description of spanoids, along with their strong connection to visible rank, to Section~\ref{subsec:spanoids}.}  to shed light on the limitations of current techniques to prove better upper bounds on the dimension of $q$-LCCs.  They noted that current techniques to bound LCC dimension apply more generally to the associated spanoids, which they showed could have rank as large as $\widetilde{\Omega}(n^{(q-2)/(q-1)})$. Therefore to improve the LCC bound one needs techniques that are more specific than spanoids and better tailored to the LCC setting. One such possibility mentioned in \cite{DGGW20} is to restrict attention to \emph{symmetric spanoids}, which have a natural symmetry property that linear LCCs imply. 
 
Our visible rank notion turns out to be intimately related to symmetric spanoids via a rank-nullity type theorem (Theorem~\ref{thm:rank-nullity}). While technically simple in hindsight, it offers a powerful viewpoint on symmetric spanoids which in particular resolves a question posed in \cite{DGGW20}---we show that symmetric spanoids are also too coarse a technique to beat the $\widetilde{O}(n^{(q-2)/(q-1)})$ upper bound on $q$-LCC dimension.

\subsection{Stencils and visible rank}
With the above backdrop, we now proceed to describe the setup we use to study these questions, based on the rank of certain matrix templates which we call "stencils." We can represent the support structure of the check equations (i.e., dual codewords) governing a locality property by an $n$-column matrix of $0$'s and $\star$'s. For each check equation involving the $i$'th symbol and a repair group $R_i \subset [n]\setminus \{i\}$, we place a row in the stencil with $\star$'s precisely at $R_i \cup \{i\}$ (i.e., with $\star$'s at the support of the associated dual codeword). For the $\ell$-LRC property for instance, an associated stencil would be an $n \times n$ matrix with $\star$'s on the diagonal and $\ell$ other $\star$'s in each row. For $q$-LCC, we would have a $\delta n^2 \times n$ matrix whose rows are split into $n$ groups with the rows in the $i$'th group having a $\star$ in the $i$'th column and $q$ other $\star$'s in disjoint columns.

The smallest co-dimension of linear codes over a field $\F$ with certain locality property is, by design, the minimum rank $\rk{\F}{H}$ of the associated stencil $H$ when the $\star$'s are replaced by arbitrary nonzero entries from $\F$. In this work, our goal is to understand this quantity via field oblivious methods based only on the combinatorial structure of the stencil of $\star$'s.

The tool we put forth for this purpose is the \emph{visible rank} of $H$, denoted $\vrk{H}$, and defined to be the largest $r$ for which there is a $r \times r$ submatrix of $H$ that has exactly one general diagonal whose entries are all $\star$'s. By the Leibniz formula, the determinant of such a submatrix is nonzero for any substitution of nonzero entries for the $\star$'s. Thus $\rk{\F}{H} \ge \vrk{H}$ for every field $\F$.

Our goal in this work is to understand the interrelationship between visible rank and the co-dimension of linear codes under various locality requirements. This can shed further light on the bottleneck in known techniques to study trade-offs between locality and code dimension, and optimistically could also lead to better constructions.

\subsection{Visible rank and Locality}

For $\ell$-LRCs, a simple greedy argument shows that its associated parity-check stencil $H$ satisfies $\vrk{H} \ge n/(\ell+1)$. Thus visible rank captures the optimal trade-off between code dimension and locality $\ell$.

For $q$-LCCs with $q \ge 3$, an argument similar to (in fact a bit simpler than and implied by) the one for spanoids in \cite{DGGW20} shows that the stencil corresponding to $q$-LCCs has visible rank at least $n - \widetilde{O}(n^{(q-2)/(q-1)})$, showing an upper bound of $\widetilde{O}(n^{(q-2)/(q-1)})$ on the dimension of $q$-LCCs. We show that visible rank suffers the same bottleneck as spanoids in terms of bounding the dimension of $q$-LCCs. 
\begin{theorem}
\label{thm:q-lcc-intro}
For $q \ge 3$, there exist $n$-column stencils $H$ with $\star$'s structure compatible with $q$-LCCs for which $\vrk{H} \le n - \widetilde{\Omega}(n^{(q-2)/(q-1)})$. 
\end{theorem}
Through the precise connection we establish between between visible rank and symmetric spanoids, this shows the same limitation for symmetric spanoids, thus answering a question posed in \cite{DGGW20}.

For the $t$-DRGP property, we focus on the $t=2$ case, with the goal of finding a combinatorial substitute for the currently known $\Omega(\sqrt{n})$ lower bounds on co-dimension~\cite{Woo16,FGW17} which are algebraic.
Unfortunately, we show that visible rank, in its basic form, is too weak in this context.
\begin{theorem}
\label{thm:2-DRGP-limits-intro}
There exist $2n\times n$ stencils $H$ with $\star$'s structure compatible with $2$-DRGP for which $\vrk{H} \le O(\log n)$. 
\end{theorem}

\subsection{Visible rank and tensor powers}

In view of Theorem~\ref{thm:2-DRGP-limits-intro}, we investigate avenues to get better bounds out of the visible rank approach. Specifically, we study the visible rank of tensor powers of the matrix. It turns out that the visible rank is super-multiplicative: $\vrk{H \otimes H} \ge \vrk{H}^2$, while on the other hand algebraic rank is sub-multiplicative, so higher tensor powers could yield better lower bounds on the rank. Indeed, we are able to show precisely this for $2$-DRGP:

\begin{theorem}
\label{thm:2-DRGP-tensor-intro}
For every $2n\times n$ stencil $H$ with $\star$'s structure compatible with $2$-DRGP, we have 
$\vrk{H \otimes H} \ge \Omega(n)$, and thus $\rk{\F}{H} \ge \Omega(\sqrt{n})$ for every field $\F$. 
\end{theorem}

On the other hand, for $q$-LCCs with $q \ge 3$, we show that higher tensor powers suffer the same bottleneck as Theorem~\ref{thm:q-lcc-intro}.
\begin{theorem}
\label{thm:q-lcc-tensor-intro}
For $q \ge 3$, there exist $n$-column stencils $H$ with $\star$'s structure compatible with $q$-LCCs for which $\vrk{H^{\otimes k}}^{1/k} \le n - \widetilde{\Omega}(n^{(q-2)/(q-1)})/k$ for any integer $k$. In particular even for $k=n^{o(1)}$, we get no polynomial improvements to the current upper bounds on dimension of $q$-LCCs. 
\end{theorem}

\subsection{Visible capacity}

Given the super-multiplicativity of visible rank under tensor powers, and drawing inspiration from the Shannon capacity of graphs, we put forth the notion of \emph{visual capacity} of a matrix $H$ of $0$'s and $\star$'s, defined as
\[ \Upsilon(H) := \sup_{k} \vrk{H^{\otimes k}}^{1/k} \ .\]
The visual capacity is also a field oblivious lower bound on algebraic rank $\rk{\F}{H}$ for any field $\F$. It is not known whether there are stencils that exhibit a gap between visible capacity and its minimum possible $\rk{\F}{H}$  over all fields $\F$. 

\smallskip
In the spirit of visual capacity, we extend Theorems~\ref{thm:2-DRGP-limits-intro} and~\ref{thm:2-DRGP-tensor-intro} to show that for any constant $t \in \N$, one can attain an exponentially better lower bound on the algebraic rank of a stencil $H$ by examining the visible rank of $H^{\otimes t}$ rather than examining the visible rank of $H^{\otimes k}$ for any $k < t$. In other words:

\begin{theorem}
\label{thm:tensor-gaps}
For any fixed natural number $t \ge 2$ and sufficiently large $n$, there exists a $tn \times n$ stencil $H$ such that $\vrk{H^{\otimes k}} = O_t((\log{n})^k)$ for any $k = 1, \ldots , t-1$ and $\vrk{H^{\otimes t}} = \Omega(n)$.
\end{theorem}

The proofs of our results are technically simple, once the framework is set up. Our contributions are more on the conceptual side, via the introduction and initial systematic study of visible rank and its diverse connections. Our inquiry also raises interesting questions and directions for future work, some of which are outlined in Section~\ref{further}, including the relationship between visible capacity and algebraic rank.

\subsection{Connections and related work}
Studying the interplay between the combinatorial structure of a matrix and its rank is a natural quest that arises in several contexts. See Chapter 3 of~\cite{Tre16} for a survey of works on lower bounding the algebraic rank. For works specific to codes with locality, the work of \cite{BDYW11} analyzed the combinatorial properties of design matrices over the reals to improve bounds on LCCs over the real numbers, although the methods used are particular to the field of reals and do not carry over to any field.

Visible rank in particular turns out to have a diverse array of connections, some of which we briefly discuss here. The connection to spanoids, that we already mentioned, is described in more detail in Section~\ref{subsec:spanoids}.

\smallskip\noindent \textbf{Uniquely restricted matchings.} 
Given a stencil $H \in \{0,\star\}^{m \times n}$, there is a \emph{canonical bipartite graph} $G$ between the rows and columns of $H$, where a row connects to a column if and only if their shared entry has a $\star$. Visual rank has a nice graph-theoretic formulation:  it turns out (see Section~\ref{subsec:vrk-combinatorics}) that a submatrix of $H$ has a unique general diagonal of $\star$'s iff the corresponding induced subgraph of $G$ has a unique perfect matching. Such induced bipartite graphs are known in the literature as \emph{Uniquely Restricted Matchings (URMs)} and have been extensively studied~\cite{GHL01, GKT01, HMT06, Mis11, TD13, FJJ18}. They were first introduced in \cite{GHL01}, wherein they proved that computing the maximum URM of a bipartite graph is \textsc{NP}-complete. It was later shown in \cite{Mis11} that $n^{1/3 - o(1)}$ approximations of the maximum URM is also \textsf{NP}-hard unless $\textsf{NP} = \textsf{ZPP}$ and additionally that the problem of finding the maximum URM is $\textsf{APX}$-complete.

\smallskip\noindent \textbf{Matroids.} 
One can encode any matroid into a stencil. Recall that a circuit of a matroid is a minimal dependent set---that is, a dependent set whose proper subsets are all independent (the terminology reflects the fact that in a graphic matroid, the circuits are cycles of the graph). Given a matroid $\mathcal{M}$ on universe $[n]$ and a set $\mathcal{C} = \{C_1,\ldots,C_m\}$ of circuits of $\mathcal{M}$, we consider a $m \times n$ stencil $H$ where the entry at $(i,j)$ is a $\star$ if and only if $j \in C_i$. For this matrix, one can show that a collection of visibly independent columns (see Section~\ref{subsec:vrk-intro} for the definition of visible independence) is an independent set in the dual matroid. Therefore, we have $\rk{}{\mathcal{M}} + \vrk{H} \le n$---this also follows from our rank-nullity theorem for symmetric spanoids as one can associate a symmetric spanoid with any matroid (the collection of sets in Definition~\ref{def:symm-spanoid} will just be the circuits of the matroid).

\smallskip\noindent \textbf{Min-rank.} 
The minimum possible rank of a square $0$-$\star$ stencil over assignments to the $\star$'s from some field has been well studied in combinatorics.\footnote{There is a slight difference in the minrank setup, in that the $\star$'s can take any value including $0$, except the $\star$'s on the diagonal which must take nonzero values.} For example, we have Haemers' classic bound on independent set of a graph and its applications to Shannon capacity~\cite{Haemers79}. Note that in this case we are using a linear-algebraic tool to understand a combinatorial quantity, whereas visible rank goes the other way, serving as a combinatorial proxy for a linear-algebraic quantity. Recent interest in minrank has included their characterization of the most efficient linear index codes~\cite{BBJK11}. The minrank of stencils corresponding to $n$-vertex random Erd\"{o}s-R\'{e}nyi graphs was recently shown to be $\Theta(n/\log n)$ over any field that is polynomially bounded~\cite{GRW18}.

\smallskip\noindent \textbf{Matrix Rigidity.}
Given a square matrix $A \in \F^{n \times n}$ and a natural number $r \le n$, the \emph{rigidity} of $A$ is the minimal number of entries that one can perturb in $A$ so that it rank becomes at most $r$. Matrix rigidity was introduced in the seminal work~\cite{Val77} and since then had expansive research on constructing explicit rigid matrices. See~\cite{Ram20} for a recent survey on matrix rigidity and related connections. The visible rank provides a combinatorial guarantee on the rank of a matrix, and that conjures up the possibility of constructing explicit rigid matrices by finding explicit stencils whose visible rank is robust to small amounts of corruptions of its entries. We pose this approach as Question~\ref{further:robust-vrk} in Section~\ref{further}.

\smallskip\noindent \textbf{Incidence Theorems.} 
Given an $m \times n$ matrix $A$ over the field $\F$ with rank $r$, one can decompose $A = MN$ where $M$ and $N$ are $m \times r$ and $r \times n$ matrices. If we consider the rows of $M$ as hyperplanes over the projective plane $\mathbb{P}\F^{r-1}$ of dimension $(r-1)$ and the columns of $N$ as points in $\mathbb{P}\F^{r-1}$, then the stencil of $A$ defines a point-hyperplane incidence over $\mathbb{P}\F^{r-1}$. In particular, when $r = 3$, the stencil of $A$ defines a point-line incidence over the field $\F$. Thus studying the combinatorial properties of a stencil whose $\F$-rank (see Definition~\ref{def:alg-rank}) is at most $3$ is equivalent to studying the combinatorics of point-line incidences over the field $\F$. For more on incidence theorems, see~\cite{Dvi12} for an excellent survey in the area.

\smallskip\noindent \textbf{Communication complexity.} The visible rank provides a connection between deterministic and nondetereministic communication complexity~\cite{lovasz-survey}. For a communication problem $f : X \times Y \to \{0,1\}$, define the stencil $H_f \in \{0,\star\}^{X \times Y}$ by $H_f(x,y) = \star$ if $f(x,y)=0$ and $M_f(x,y) = 0$ if $f(x,y)=1$. Then it is known that $D(f) \le (\log_2 \vrk{H_f}) \cdot (N(f) + 1)$ where $D(f)$ and $N(f)$ are respectively the deterministic and nondeterministic communication complexity of $f$~\cite[Thm 3.5]{lovasz-survey}.

\smallskip\noindent \textbf{Distance of expander codes.} We can associate a stencil with the parity-check matrix $H$ of a code $C$ in the natural way, by placing $\star$'s at the positions of the nonzero field elements. Suppose that every subset of $d$ columns of the stencil is visibly independent. Then the minimum distance of $C$ is at least $d$. This is in fact the lower bound on distance of expander codes established in \cite{SS96}. The visible independence of subsets of $d$ columns is argued via \emph{unique expansion} of the canonical bipartite graph associated with $H$ (also called the factor graph in coding theory). Namely, for every subset $S$ of $d$ vertices on the left, there is a vertex on the right that's adjacent to exactly one vertex in $S$. Thus, linear independence is really argued via visible independence.

\subsection{Organization} 
 We begin in Section~\ref{stencil} by formally introducing the notations and terminology for stencils, and  establishing some simple but very useful combinatorial facts about visible rank. We use these to show that there are $q$-LCC stencils for $q \ge 3$ with visible rank at most $n - \widetilde{\Omega}(n^{(q-2)/(q-1)})$ (Section~\ref{lcc}), and the existence of a $2$-DRGP stencil with visible rank of at most $O(\log{n})$ (Section~\ref{drgp}). In Section~\ref{tensor}, we introduce a tensor product operation on stencils and prove various properties about them. In Section~\ref{locality-tp}, we utilize tensor powers to show that the rank of a $2$-DRGP over any field $\F$ is at least $\sqrt{n}$, which asymptotically matches the current best lower bounds on $t$-DRGP codes. We also show that for $q$-LCC stencils, the tensor powers at the $k$'th level for $k \le \mathrm{polylog}(n)$ do not yield better lower bounds on the rank than the ones obtained from the visible rank. Finally, in Section~\ref{further}, we discuss further directions and questions inspired by this work.

\section{Stencils and their visible rank}
\label{stencil}
In this section, we will be formally setting up the model of stencils and all the associated definitions and notations.

We denote $[n]$ to be the set $\{1, 2, \ldots , n\}$. For any matrix $H \in \{0, \star\}^{m \times n}$, we denote it as a \emph{stencil}. For an $m \times n$ stencil $H$, we denote its entry in the $i$'th row and $j$'th column by $H[i,j]$. Any restriction to the specific sub-collection of the rows and columns of $H$ is said to be a \emph{sub-stencil} of $H$. For given sets $A$ and $B$, a stencil $H$ is said to be an $A \times B$ if it is an $|A| \times |B|$ stencil along with an associated indexing of the rows by $A$ and the columns by $B$. Given a square stencil $M \in \{0, \star\}^{n \times n}$, a \emph{general diagonal} of $M$, is a collection of entries $\{M[1,\pi(1)], \ldots , M[n,\pi(n)]\}$ where $\pi$ is a permutation on $[n]$. We say that a general diagonal is a \emph{star diagonal} if all its $n$ entries are $\star$'s.

\subsection{Algebraic witnesses of stencils}
Instantiating a code with the locality properties stipulated by a stencil amounts to filling its $\star$'s with field entries, or realizing an algebraic witness as defined below.
\begin{defn}[Algebraic witness]
\label{def:alg-wit}
For field $\F$ and stencil $H \in \{0, \star\}^{m \times n}$, a matrix $W \in \F^{m \times n}$ is said to be an $\F$-witness of $H$ if it satisfies the property that $W[i,j] \neq 0$ if and only if $H[i,j] = \star$. More generally, any $\F$-witness of $H$ is said to be an algebraic witness of $H$.
\end{defn}

We stress that every $\star$ in the stencil $H$ must be replaced by a \emph{nonzero} entry from $\F$ and cannot be zero. Of the possible algebraic witnesses for $H$, we will be primarily focused in this paper on the algebraic witnesses that attain the smallest feasible rank, which leads us to the following definition.

\begin{defn}[Rank]
\label{def:alg-rank}
Given an $m \times n$ stencil $H$, the $\F$-rank of $H$ is the smallest natural number $r$ such that there exists an $\F$-witness $W \in \F^{m \times n}$ whose rank is equal to $r$. We denote the value $r$ by $\rk{\F}{H}$.
\end{defn}

\subsection{Visible Rank}
\label{subsec:vrk-intro}
In this section, we introduce our notion of the \emph{visible rank} of a stencil. The main motivation of introducing the visible is to be able to determine the most optimal lower bound on the rank of a matrix with only the knowledge of knowing the support of a matrix and nothing else about the values of that support. 

Consider a square matrix $A \in \F^{n \times n}$, and suppose we are interested in determining if it is full rank. A natural approach would be to inspect its determinant. From the Leibniz formula, we know that $\det(A) = \sum_{\pi \in S_n}{(-1)^{\text{sgn}(\pi)}\prod_{i=1}^n{A_{i, \pi(i)}}}$, where $S_n$ denotes the symmetric group of order $n$ and $\text{sgn}(\pi)$ denotes the sign of a permutation $\pi$. From the Leibniz formula, notice that $\det(A)$ is a linear combination of the nonzero general diagonals of $A$. If our hope is to obtain $\det(A) \neq 0$ without any knowledge of the values of the support of $A$, then one way to guarantee it is to say that $A$ has exactly one nonzero general diagonal. In such a case, we can guarantee that $\det(A) \neq 0$. As when $A$ has more than one general diagonal, there is no guarantee if $\det(A) \neq 0$ without inspecting the values of the support of $A$.

From the previous discussion, it seems natural to define the notion of a rank on stencils as follows.

\begin{defn}[Visibly Full Rank]
\label{def:vis-full-rank}
For a square stencil $M \in \{0, \star\}^{n \times n}$, we say that $M$ is visibly full rank if $M$ has exactly one $\star$ diagonal. That is, a general diagonal whose entries are all $\star$'s.
\end{defn}

Of course, in most cases, when we are given a matrix $A \in \F^{m \times n}$, we would be interested in determining its rank. One way to define the rank of the matrix $A$ is to say that $\text{rank}(A)$ is the size of the largest square submatrix in $A$ that is full-rank. From this viewpoint, it seems clear to define the rank of a stencil in a similar fashion.

\begin{defn}[Visible Rank]
\label{def:vrk}
For a stencil $H \in \{0, \star\}^{m \times n}$, the visible rank of $H$, denoted $\vrk{H}$, is the largest square sub-stencil in $H$ that is visibly full rank.
\end{defn}

We also say that a set of $k$ columns in $H$ is \emph{visibly independent} if there exists a $k \times k$ sub-stencil within these $k$ columns that is visibly full rank.

Of course, not all full-rank square matrices $A \in \F^{n \times n}$ necessarily have exactly one nonzero general diagonal, but all square matrices that have exactly one nonzero general diagonal are necessarily full-rank. Thus if we are interested in determining $\text{rank}(A)$ by finding the size of the largest square submatrix in $A$ that is full-rank, we can instead search for the largest square submatrix in $A$ that has exactly one nonzero general diagonal, and that will lead us to lower bounds on the rank of $A$.

\begin{prop}
\label{prop:ark-lb}
Given a field $\F$ and stencil $H \in \{0, \star\}^{m \times n}$, we have the inequality $\rk{\F}{H} \ge \vrk{H}$.
\end{prop}
\begin{proof}
Let $W \in \F^{m \times n}$ be algebraic witness of $H$. Let $H_0$ be a $k \times k$ square sub-stencil of $H$ that is visibly independent. Let $W_0$ be the corresponding $k \times k$ submatrix of $H_0$ in $W$. Because $H_0$ has exactly one general diagonal of $\star$'s, then $W_0$ has exactly one general diagonal of nonzero entries, which implies $\det(W_0) \neq 0$. This proves that $\rank(W) \ge \rank(W_0) = k$. By picking the largest such $H_0$, we deduce that $\rank(W) \ge \vrk{H}$. Since $W$ was an arbitrary $\F$-witness , this proves $\rk{\F}{H} \ge \vrk{H}$.
\end{proof}

\subsection{Combinatorial properties of visible rank}
\label{subsec:vrk-combinatorics}
In this subsection, we will be proving some properties about visible rank. In particular, we will show that any visibly full rank square stencil $M$ is permutationally equivalent to an upper triangular stencil, and from this observation, we will be able to upper bound the visible rank by the largest rectangle of zeros in the stencil, which will be our main tool in our constructions of $q$-LCC and $t$-DRGP stencils. We also show an upper bound on the $\F$-rank of a stencil by the maximum number of zeros in each row over any field $\F$ satisfying $\abs{\F} \ge n$ (where $n$ is the number of columns).

Given two stencils $H_1, H_2 \in \{0, \star\}^{m \times n}$, we say that $H_1$ is \emph{permutationally equivalent} to $H_2$ if there are permutations $\pi : [m] \to [m]$ and $\sigma : [n] \to [n]$ such that $H_1[i,j] = H_2[\pi(i), \sigma(j)]$ for all $i \in [m]$ and $j \in [n]$. For such $H_1$ and $H_2$, we introduce the notation $H_2 = (H_1)_{\pi, \sigma}$ to say that $H_2$ is obtained from $H_1$ by permuting the rows with the permutation $\pi$ and the columns by the permutation $\sigma$ (We remark that row permutations commute with column permutations). 

\begin{lemma}
\label{lemma:upper-tri}
Let $M \in \{0, \star\}^{n \times n}$ be visibly full rank. Then there exists permutations $\pi$ and $\sigma$ on $[n]$ such that $N \coloneqq M_{\pi, \sigma}$ is an upper triangular stencil. That is $N[i,i] = \star$ and $N[i,j] = 0$ for all $i,j \in [n]$ with $i > j$.
\end{lemma}
\begin{proof}
We first prove the following claim.
\begin{claim}
\label{claim:leaf}
Let $M \in \{0, \star\}^{n \times n}$ be visibly full rank. Then there exists a row in $M$ with exactly one $\star$.
\end{claim}
\begin{claimproof}
Assume (for the sake of a contradiction) that such a row doesn't exist. Index the rows of $M$ by $R = \{r_1, \ldots , r_n\}$ and the columns by $C = \{c_1, \ldots , c_n\}$. Let $G = (R,C,S)$ be a bipartite graph on the rows and columns of $M$ with edges $S$, where $S$ is the set of $\star$'s in $M$. Because $M$ is visibly independent, $G$ has a unique perfect matching. Moreover, by our initial assumption, $d_G(v) \ge 2$ for all $v \in L$.

Color the edges of the unique matching of $G$ red and all remaining edges of $G$ blue. Let $P \subseteq S$ be the maximum alternating path in $G$. Since every vertex in $G$ is incident to a red edge, then by the maximality of $P$, both edges at the ends of $P$ must be red. Let $v_1$ be a vertex at one of the ends of $P$. Since $d_G(v_1) \ge 2$ and $v_1$ is incident to only one red edge, then $v_1$ is incident to a blue edge, but because our path is maximal, then this blue edge $(v_1, v_2)$ must satisfy $v_2 \in P$. This implies that our graph $G$ has an alternating cycle $C$ with red edges $R_C$ and blue edges $B_C$. Notice that both $R_C$ and $B_C$ match the same vertex sets, and since $M$ is a matching, then so is $(M \setminus R_C) \cup B_C$, but that's a contradiction by the uniqueness of $M$.
\end{claimproof}

We now proceed to our lemma by induction on $n$. The base case $n=1$ is immediate to see. As for the induction step, index the rows and columns of $M$ by $\{1, \ldots , n\}$. We know by Claim~\ref{claim:leaf} that there exists a row $i$ in $M$ which has exactly a single $\star$ in column $j$. Let $M'$ be the $n \times n$ matrix produced from $M$ by transposing row $i$ with row $n$ and column $j$ with column $n$. Because $M$ is visibly independent, so is $M'$. Moreover, the $n$'th row of $M'$ has exactly one $\star$ at column $n$, meaning that any $\star$ diagonal in $M'$ must contain $M[n,n]$. Thus the $(n-1) \times (n-1)$ minor $M'_0$, which is obtained by deleting row $n$ and column $n$ of $M'$, has the same number of general diagonals as $M'$. This means $M'_0$ is visibly independent. By our induction hypothesis, we can permute the rows and columns of $M'_0$ to make it upper triangular, and thus the same holds for $M'$. Since $M'$ is a permutation of $M$, we conclude that $M$'s rows and columns can be permuted to make it upper triangular. This proves our induction step.
\end{proof}

Thus we can characterize all visibly full rank stencils, and that help us obtain the following upper bound on the visible rank.

\begin{lemma}
\label{lemma:biclique-bound}
Given an $m \times n$ stencil $H$, if there are natural numbers $a,b$ such that $H$ has no $a \times b$ sub-stencil of zeros, then $\vrk{H} < a+b$.
\end{lemma}
\begin{proof}
Assume (for the sake of a contradiction) that $H$ has a $(a+b) \times (a+b)$ square sub-stencil $H_0$ that is visibly independent. By Lemma~\ref{lemma:upper-tri}, we know that $H_0$ is permutationally equivalent to an $(a+b) \times (a+b)$ upper triangular stencil. Since such an upper triangular stencil has a $a \times b$ sub-stencil of zeros, we arrive to a contradiction.
\end{proof}

We also provide an upper bound on the rank of a stencil by the maximum number of zeros in each rows.

\begin{prop}
\label{prop:zero-set-ub}
For any $m \times n$ stencil $H$, if each row of $H$ has at most $d$ zeros, then $\rk{\F}{H} \le d+1$ for all fields $\F$ such that $|\F| \ge n$.
\end{prop}
\begin{proof}
Since $|\F| \ge n$, then we can label the columns of $H$ by pairwise distinct entries $a_1, \ldots , a_n \in \F$. For row $i$, let the columns that are zero along row $i$ be $Z_i \subseteq \{a_1, \ldots , a_n\}$. Consider the polynomial $p_i(x) \coloneqq \prod_{a \in Z_i}{(x-a)}$. Notice that $p_i$ evaluates to zero on $Z_i$, while on $\{a_1, \ldots , a_n\} \setminus Z_i$, it evaluates to nonzero values. This means that the matrix $E \in \F^{m \times n}$ defined by $E_{ij} = p_i(a_j)$ is an $\F$-witness of $H$. Moreover, since $|Z_i| \le d$ for each $i \in [m]$, then we know that the monomials $\{1, x, \ldots , x^d\}$ span the polynomials $\{p_1, \ldots , p_m\}$. This shows that $\rank(E) \le d+1$ and thus $\rk{\F}{H} \le d+1$.
\end{proof}

We remark that the bound $\abs{\F} \ge n$ is crucial for Propsition~\ref{prop:zero-set-ub}. Consider the stencil $D \in \{0, \star\}^{n \times n}$ that has $\star$'s everywhere except on the diagonal. Such a stencil has a visible rank of $2$, but its rank over $\F_2$ is at least $n-1$.

\subsection{A Rank-Nullity Type Theorem Between Stencils and Symmetric Spanoids}
\label{subsec:spanoids}
In this subsection, we formally setup spanoids and prove a rank-nullity type theorem between symmetric spanoids and stencils.\smallskip

A \emph{spanoid} $\mathcal{S}$ is a collection of inference rules in the form of pairs $(S,i)$, which are written in the form $S \to i$, where $S \subseteq [n]$ and $i \in [n]$. The objective in spanoids is to determine the size of the smallest subset $B \subseteq [n]$ such that one can use the inference rules of $\mathcal{S}$ to obtain all of $[n]$. Spanoids were introduced in \cite{DGGW20} as an abstraction of LCCs, wherein they proved that the spanoid analog of $q$-LCCs satisfy the upper bound $\widetilde{O}(n^{(q-2)/(q-1)})$ on the rank. Moreover, they also showed that there are $q$-LCC spanoids for which their rank is $\widetilde{\Omega}(n^{(q-2)/(q-1)})$.

Let us setup the definitions needed for spanoids. A derivation in $\mathcal{S}$ of $i \in [n]$ from a set $T \subseteq [n]$ is a sequence of sets $T_0=T, T_1, \dots, T_r$ satisfying $T_j = T_{j-1} \cup \{i_j\}$ for some $i_j \in [n]$, $j \in [r]$, and with $i_r = i$. Further, for every $j \in [r]$, there is a rule $(S_{j-1},i_j)$ in $\mathcal{S}$ for some $S_{j-1} \subseteq T_{j-1}$. The \emph{span} of a set $T \subseteq [n]$, denoted $\mathsf{span}_{\mathcal{S}}(T)$, is the set of all $i \in [n]$ for which there is a derivation of $i$ from $T$. The rank of a spanoid, denoted $\mathsf{rank}(\mathcal{S})$, is the size of the smallest set $T \subseteq [n]$ such that $\mathsf{span}_{\mathcal{S}}(T) =[n]$. Finally, we define symmetric spanoids below.

\begin{defn}[Symmetric Spanoids]
\label{def:symm-spanoid}
A spanoid $\mathcal{S}$ over $[n]$ is a symmetric spanoid if there is a collection of sets $\{S_1, \ldots , S_m\}$ so that the inference rules of $\mathcal{S}$ are of the form $S_j \setminus \{i\} \to \{i\}$ for any $i \in S_j$ and $j \in [m]$.
\end{defn}

Now we may proceed to prove our theorem that relates the rank of symmetric spanoids with the visible rank of an associated stencil. 

\begin{theorem}
\label{thm:rank-nullity}
For any symmetric spanoid $\mathcal{S}$ over $[n]$ with $m$ sets, there exists a canonical stencil $H$ of size $m \times n$ such that for any collection of columns $C \subseteq [n]$ in $H$, they are visibly independent if and only if $\mathsf{span}_{\mathcal{S}}([n] \setminus C) = [n]$. Moreover, we have $\vrk{H} + \mathsf{rank}(\mathcal{S}) = n$.
\end{theorem}
\begin{proof}
Define $H[i,j] = \star$ if $j \in S_i$ and zero otherwise. We claim that such $H$ satisfies the conditions. Indeed, suppose that the columns $C = \{c_1, \ldots , c_k\}$ are visibly independent. Then that means there are rows $r_1, \ldots , r_k$ so that the $k \times k$ sub-stencil formed by these columns and rows is visibly full rank. Denote this matrix as $H_C$. By Lemma~\ref{lemma:upper-tri}, we can find permutations $\pi, \sigma$ over $[k]$ such that the matrix $H_C' \coloneqq H_{\pi, \sigma}$ is upper triangular. In terms of spanoids, that means $c_{\sigma(i)} \in S_{r_{\pi(i)}}$ and $S_{r_{\pi(i)}} \subseteq [n] \setminus \{c_{\sigma(1)}, \ldots , c_{\sigma(i-1)}\}$ for all $i \in [k]$. We can rewrite the last set containment as $S_{r_{\pi(i)}} \subseteq ([n] \setminus C) \cup \{c_{\sigma(i)}, \ldots , c_{\sigma(k)}\}$. Thus we apply the inference rules $S_{r_{\pi(k)}} \setminus c_{\sigma(k)} \to c_{\sigma(k)}, S_{r_{\pi(k-1)}} \setminus c_{\sigma(k-1)} \to c_{\sigma(k-1)}, \ldots , S_{r_{\pi(1)}} \setminus c_{\sigma(k)} \to c_{\sigma(1)}$ in that order with the set $[n] \setminus C$, to deduce that the set $[n] \setminus C$ spans the set $[n]$ in $\mathcal{S}$. Thus $\mathsf{span}_{\mathcal{S}}([n] \setminus C) = [n]$.

Now, suppose $\mathsf{span}_{\mathcal{S}}([n] \setminus C) = [n]$. Then that means we can find sets $S_{i_1}, \ldots , S_{i_k}$ and a permutation $\sigma$ over $[k]$ such that $c_{\sigma(j)} \in S_{i_j}$ and we can apply the inference rules $S_{i_1} \setminus \{c_{\sigma(1)}\} \to c_{\sigma(1)}, \ldots , S_{i_k} \setminus \{c_{\sigma(k)}\} \to c_{\sigma(k)}$ in that order to span $[n]$ from $[n] \setminus C$ in $\mathcal{S}$. That implies then that $S_{i_j} \subseteq ([n] \setminus C) \cup \{c_{\sigma(1)}, \ldots , c_{\sigma(j)}\}$. In terms of the stencil $H$, that means $H[i_j, \sigma(j)] = \star$ and $H[i_{\ell}, \sigma(j)] = 0$ for $\ell < j$. Thus the $k \times k$ sub-stencil $H'$ obtained by restricting to the columns $c_{\sigma(1)}, \ldots , c_{\sigma(k)}$ and rows $i_1, \ldots , i_k$ in that order forms a lower triangular stencil, which is permutationally equivalent to an upper triangular stencil. Thus we deduce that the set of columns $C$ is visibly independent.

Now, for any set $S \subseteq [n]$ such that $\mathsf{span}_{\mathcal{S}}(S) = [n]$, we know that $[n] \setminus S$ is visibly independent in $H$. Thus $n - \abs{S} \le \vrk{H}$. Since this holds for any such set $S$, then we deduce that $\mathsf{rank}(\mathcal{S}) + \vrk{H} \ge n$. On the other hand, for any collection of columns $C$ in $H$ that is visibly independent, we know that $\mathsf{span}_{\mathcal{S}}([n] \setminus C) = [n]$. This implies $n - \abs{C} \ge \mathsf{rank}(\mathcal{S})$. Since this holds for any visibly independent set of columns $C$ in $H$, then we find that $n \ge \mathsf{rank}(\mathcal{S}) + \vrk{H}$. Hence $\mathsf{rank}(\mathcal{S}) + \vrk{H} = n$.
\end{proof}

\section{Constructing $q$-LCC Stencils}
\label{lcc}
In this section, we define $q$-LCC stencils and construct $q$-LCC stencils whose visible rank achieves the known lower bounds up to polylog factors.

\begin{defn}[$q$-LCC Stencils]
\label{def:lcc-stencil}
For $\delta > 0$, a $\delta n^2 \times n$ stencil $H$ whose rows are labelled by $[n] \times [\delta n]$ is said to be a $q$-LCC stencil if $H[(i,j),i] = \star$ for all $(i,j) \in [n] \times [\delta n]$. Moreover, for every $k \in [n] \setminus \{i\}$, the collection of entries \{$M[(i,1),k], \ldots , M[(i,t),k]\}$ has at most one $\star$ (where $t = \delta n$), and the number of $\star$'s in each row of $H$ is at most $q+1$.
\end{defn}

Now we proceed to prove our main theorem for this section.

\begin{theorem}
\label{thm:lcc-construction}
For $q \ge 3$, there exists a $q$-LCC stencil $M$ whose visible rank is at most $n - \widetilde{\Omega}\left(n^{(q-2)/(q-1)}\right)$.
\end{theorem}
\begin{proof}
For $(i,j) \in [n] \times [\delta n]$. Define $r^i_j \coloneqq \{k \in [n] : H[(i,j), k] = \star\}$ to be the support of row $(i,j)$, and let $G_i \coloneqq \{r^i_j : j \in [\delta n]\}$ be the $\delta n$ groups for column $i$. We shall show that by picking the groups $G_i$ uniformly at random, the visible rank will at most be $n - n^{(q-2)/(q-1)} / \log{n}$ with high probability.

Consider natural numbers $s > k$ where $s = n^{\frac{q-2}{q-1}}$ and $k = s/\log{n} = n^{\frac{q-2}{q-1}}/\log{n}$. Let $E_{s,k}$ be the event that there aren't any $(s-k) \times (n-s)$ sub-stencils in $H$ that are all zeros. We shall show that $E_{s,k}$ occurs with high probability, which will yield an upper bound of $n-k$ on the visible rank via Lemma~\ref{lemma:biclique-bound}. We will use an equivalent form of the event $E_{s,k}$, which is the event that there are no $s-k$ rows in $H$ whose union of supports is at most $s$. Now, consider a collection of columns $C$ of size $s$ and a collection of $s-k$ rows $R$, where $R$ and $C$ denote the collection of their supports. Enumerate $R \cap G_i = \{r^i_1, \ldots , r^i_{a_i}\}$. By the chain rule, the definition of $G_i$, and the independence of the $G_i$'s, we find that
\begin{align*}
\Pr{\bigwedge_{r \in R}{r \subseteq C}} &= \prod_{i=1}^n{\Pr{\bigwedge_{r \in R \cap G_i}{r \subseteq C}}} \\
&= \prod_{i=1}^n{\prod_{j=1}^{a_i}{\Pr{r^i_j \subseteq C \; \bigg| \; r^i_k \subseteq C \text{ for $k \in [j-1]$}}}} \\
&= \prod_{i=1}^n{\prod_{j=1}^{a_i}{\Pr{r^i_j \setminus \{i\} \subseteq C \setminus \{r^i_1, \ldots , r^i_{j-1}\} \; \bigg| \; r^i_k \subseteq C \text{ for $k \in [j-1]$}}}} \\
&= \prod_{i=1}^n{\prod_{j=1}^{a_i}{\frac{\binom{s-(j-1)q-1}{q}}{\binom{n-(j-1)q-1}{q}}}} \le \left(\frac{\binom{s-1}{q}}{\binom{n-1}{q}}\right)^{s-k} \le \left(\frac{s}{n}\right)^{q(s-k)}
\end{align*}
Therefore, from the definition of $E_{s,k}$, by applying a Union Bound over all possible collections of columns $C$ of size $s$ and $s-k$ collections of rows $R$ (note that $R$ must be from the $\delta ns$ rows that form the groups of $C$ for them to be contained in $C$) and use the bound $\binom{a}{b} \le \left(\tfrac{ea}{b}\right)^b$, we deduce that
\begin{align*}
\Pr{E_{s,k}} \le \binom{n}{s}\binom{\delta ns}{s-k}\left(\frac{s}{n}\right)^{q(s-k)} &\le \left(\frac{en}{s}\right)^s\left(\frac{e\delta ns}{s-k}\right)^{s-k}\left(\frac{s}{n}\right)^{q(s-k)} \\
&= \left(\frac{en}{s}\right)^k\left(\frac{e^2\delta s^{q-1}}{\left(1-\frac{k}{s}\right)n^{q-2}}\right)^{s-k} \\
&= \left(en^{1/(q-1)}\right)^k\left(\frac{e^2\delta}{\left(1-\frac{1}{\log{n}}\right)}\right)^{n^{\frac{q-2}{q-1}}\left(1-\frac{1}{\log{n}}\right)}
\end{align*}
Thus for small enough $\delta < e^{-2}$, the quantity above becomes $\exp\left(-\Omega(n^{(q-2)/(q-1)})\right)$. Thus we can find a $q$-LCC stencil $M$ such that no $s-k$ rows whose support is entirely contained within $s$ columns. That is equivalent to saying that there is no $(s-k) \times (n-s)$ sub-stencil that is all zeros. By Lemma~\ref{lemma:biclique-bound}, we therefore conclude $\vrk{M} < n-k = n -\Omega\left(n^{\frac{q-2}{q-1}}/\log{n}\right)$.
\end{proof}

\section{$t$-DRGP Stencils}
\label{drgp}
We now define the stencils that capture the requirement of each codeword symbol having $t$ disjoint recovery groups.
\begin{defn}[$t$-DRGP Stencils]
\label{def:drgp-stencil}
A $tn \times n$ stencil $H$ whose rows are labelled by $[n] \times [t]$ is said to be a $t$-DRGP stencil if $H[(i,j),i] = \star$ for all $(i,j) \in [n] \times [t]$. Moreover, for every $k \in [n] \setminus \{i\}$, the collection of entries \{$M[(i,1),k], \ldots , M[(i,t),k]\}$ has at most one $\star$.
\end{defn}

Now we proceed to prove our main theorem for this section.

\begin{theorem}
\label{thm:drgp-gap}
For any fixed natural number $t \ge 2$, there exists a $t$-DRGP stencil $H$ satisfying $\vrk{H} \le O(t^2 \log{n})$.
\end{theorem}
\begin{proof}
Consider a random $t$-DRGP stencil $H$ as follows: define the set of entries $S_{i,j} \coloneqq \{(i,s),j) \; | \; s \in [t]\}$. For each $i \neq j \in [n]$, set all of the entries $S_{i,i}$ to be $\star$'s, and uniformly sample an entry from $S_{i,j}$ to be a $\star$, while everything else in $S_{i,j}$ is set to  zero.

We will show that $\vrk{H} \le c_t\log{n}$ occurs with high probability, for some constant $c_t > 0$ depending on $t$. Indeed, fix $k \in \N$. Given any square sub-stencil $H_0$ of $H$ of size $k$, we have by Lemma~\ref{lemma:upper-tri} that if $H_0$ is visibly independent, then it must have at least $\binom{k}{2}$ zero entries. Let this set of entries be $Z \subseteq ([n] \times [t]) \times [n]$. Since $Z$ must all be zeros, then $Z \subseteq \cup_{i \neq j}{S_{i,j}}$. Since each $S_{i,j}$ has size $t$, then $Z$ has at least $\binom{k}{2}/t$ entries, each of which belongs to an $S_{i,j}$ that is different than the other. That is, for each $i \neq j \in [n]$ such that $Z \cap S_{i,j} \neq \varnothing$, pick an arbitrary entry $e \in Z \cap S_{i,j}$, and let $T$ be those set of entries. Then we know that $\abs{T} \ge \binom{k}{2}/t$. Moreover, the events that the entries in $T$ are zero are all independent, with each having a chance of at most $1 - 1/t$ of being zero. Thus the chance that $Z$ is all-zeros is at most $(1 - 1/t)^{\binom{k}{2}/t}$. By a Union Bound over all possible such $Z$'s, which is enumerated over all $(k!)^2$ different permutations of the rows and columns, we deduce that
\begin{equation*}
    \Pr{\text{$H_0$ is visibly independent}} \le (k!)^2\left(1 - \frac{1}{t}\right)^{\frac{\binom{k}{2}}{t}} \le \left(k^2\left(1 - \frac{1}{t}\right)^{\frac{k-1}{2t}}\right)^k \ .
\end{equation*}
And by applying another Union Bound over all such $H_0$, we find that
\begin{equation*}
    \Pr{\vrk{H} \ge k} \le \binom{tn}{k}\binom{n}{k}\left(k^2\left(1 - \frac{1}{t}\right)^{\frac{k-1}{2t}}\right)^k \le \left(k^2tn^2\left(1 - \frac{1}{t}\right)^{\frac{k-1}{2t}}\right)^k \ .
\end{equation*}
Picking $k = 6t^2\ln{n} + 1$ makes the right hand side less than $1$ for large enough $n$.
\end{proof}

\section{Tensor Products}
\label{tensor}

In this section, we introduce a tensor product operation on stencils and then proceed to explore its properties. Given the natural tensor product $A \otimes B$ for matrices $A$ and $B$, notice that the support of $A \otimes B$ is determined completely by the support of the matrices $A$ and $B$. As a consequence of this observation, we are able to define the stencil of $A \otimes B$ based solely on the stencils of $A$ and $B$. This leads us to our definition of a tensor product between stencils.

\begin{defn}[Tensor product]
\label{def:tensor-prod}
Given an $A_1 \times B_1$ stencil $H_1$ and an $A_2 \times B_2$ stencil $H_2$, let $H_1 \otimes H_2$ be a $(A_1 \times A_2) \times (B_1 \times B_2)$ stencil defined by
\begin{equation*}
    (H_1 \otimes H_2)[(a_1, a_2), (b_1, b_2)] = 
    \begin{cases}
        \star \quad &\text{if $H_1[a_1, b_1]$ and $H_2[a_2, b_2]$  both equal $\star$,} \\
        0 \quad &\text{if at least one of $H_1[a_1, b_1]$ and $H_2[a_2, b_2]$ equals $0$.}
    \end{cases}
\end{equation*}
\end{defn}

We remark that our tensor product follows similar properties as the natural tensor product for matrices, such as associativity and non-commutativity.

\subsection{Algebraic witnesses of tensor products}

In this subsection, we prove that the tensor product of any algebraic witnesses of the stencils $H_1$ and $H_2$ is an algebraic witness of $H_1 \otimes H_2$. This therefore shows us that the $\F$-rank is a \emph{sub-multiplicative} function with respect to the tensor product.

\begin{prop}
\label{prop:aw-tensor}
Let $M$ and $N$ be matrices over a field $\F$ who are $\F$-witnesses to stencils $H_1, H_2$, respectively. Then $M \tensor N$ is an $\F$-witness of $H_1 \tensor H_2$.
\end{prop}
\begin{proof}
For every entry in $H_1 \tensor H_2$, we know that $(H_1 \tensor H_2)[(i_1,i_2),(j_1,j_2)]$ is a $\star$ if and only if $H_1[i_1,j_1]$ and $H_2[i_2,j_2]$ are both $\star$'s. This holds if and only if $M_{i_1j_1}$ and $N_{i_2j_2}$ are both nonzero. Because $(M \tensor N)_{(i_1,i_2),(j_1,j_2)} = M_{i_1j_1}N_{i_2j_2}$, then the entry $(M \tensor N){(i_1,i_2),(j_1,j_2)} $ is nonzero if and only if $M_{i_1j_1}$ and $N_{i_2j_2}$ are both nonzero. Thus $M \tensor N$ is an $\F$-witness of $H_1 \tensor H_2$.
\end{proof}

By applying Proposition~\ref{prop:aw-tensor} on the $\F$-witnesses of $H_1$ and $H_2$ with the smallest ranks, we deduce the following corollary.

\begin{corollary}
\label{cor:ark-sub-mult}
For any field $\F$, we have the inequality $\rk{\F}{H_1}\rk{\F}{H_2} \ge \rk{\F}{H_1 \tensor H_2}$
\end{corollary}

\subsection{Visible rank and tensor products}

In this subsection, we show that the tensor product of two visibly full rank stencils is also visibly full rank. This therefore shows that the visible rank is a \emph{super-multiplicative} function with respect to the tensor product. We also show an upper bound on the visible rank of the tensor product with respect to the visible rank of one of the stencils.

\begin{prop}
\label{prop:vi-tensor}
Given visibly independent matrices $A$ and $B$ of size $n$, their tensor $A \tensor B$ is also visibly independent.
\end{prop}
\begin{proof}
By Lemma~\ref{lemma:upper-tri}, we know that there are permutations $\pi_A, \sigma_A, \pi_B, \sigma_B$ on $[n]$ such that the stencils $A_0 = (A)_{\pi_A, \sigma_A}$ and $B_0 = (B)_{\pi_B, \sigma_B}$ are both upper triangular stencils. Moreover, we see that $A_0 \tensor B_0 = (A \tensor B)_{(\pi_A,\pi_B), (\sigma_A,\sigma_B)}$. Therefore, it suffices for us to show that $A_0 \tensor B_0$ is an upper triangular stencil.

Consider the lexicographical ordering on $[n] \times [n]$. When $(i_1, i_2) > (j_1, j_2)$, we know that one of the inequalities $i_1 > j_1$ and $i_2 > j_2$ must hold, which means one of $A_0[i_1, j_1]$ or $B_0[i_2, j_2]$ must be a zero entry. This proves that $(A_0 \tensor B_0)[(i_1,i_2),(j_1,j_2)] = 0$ whenever $(i_1, i_2) > (j_1, j_2)$. As for when $(i_1, i_2) = (j_1, j_2)$, we immediately know that $(A_0 \tensor B_0)[(i_1,i_2),(i_1,i_2)] = \star$ as $A_0[i_1, i_1] = \star$ and $B_0[i_2, i_2] = \star$. Hence $A_0 \tensor B_0$ is an upper triangular stencil with respect to the lexicographical ordering.
\end{proof}

Given stencils $H_1$ and $H_2$, we know by Proposition~\ref{prop:vi-tensor} that the tensor product of any of their visibly full rank sub-stencils will also be visibly full rank in $H_1 \tensor H_2$. This yields the following corollary.

\begin{corollary}
\label{cor:vrk-super-mult}
For stencils $H_1$ and $H_2$, We have the inequality $\vrk{H_1 \tensor H_2} \ge \vrk{H_1}\vrk{H_2}$.
\end{corollary}

Lastly, we end this subsection with an upper bound on the visible rank of $H_1 \tensor H_2$.

\begin{prop}
\label{prop:vrk-tensor-ub}
For stencils $H_1$ and $H_2$ of sizes $m_1 \times n_1$ and $m_2 \times n_2$, respectively, we have the inequality $\vrk{H_1 \tensor H_2} \le \vrk{H_1}n_2$.
\end{prop}
\begin{proof}
Consider a visibly full rank sub-stencil $M$ in $H_1 \otimes H_2$ of size $k \times k$. By Lemma~\ref{lemma:upper-tri}, we can find permutations $\pi$ and $\sigma$ on the rows and columns of $M$, respectively, so that the stencil $M_0 \coloneqq M_{\pi, \sigma}$ is an upper triangular stencil. Let the columns of $M_0$ be indexed from  left to right by $(a_1, b_1), \ldots , (a_k, b_k)$. Let the rows of $M_0$ be indexed from top to bottom by $(c_1, d_1), \ldots , (c_k, d_k)$. For $b \in [n_2]$, define $I_b \coloneqq \{\ell \in [k] \; : \; b_\ell = b\}$. Define $b_{max} \coloneqq \arg\max_{b \in [n_2]}\{\abs{I_b}\}$, and let $s \coloneqq \abs{I_{b_{max}}}$. Then from these definitions, it follows that $s \ge k / n_2$. 

Now, write $I_{b_{max}} = \{i_1 < \ldots < i_s\}$, and consider the $s \times s$ sub-stencil $M_1$ in $M_0$ with columns $(a_{i_1}, b_{i_1}), \ldots , (a_{i_s}, b_{i_s})$ and rows $(c_{i_1}, d_{i_1}), \ldots , (c_{i_s}, d_{i_s})$. Since $M_0$ is an upper triangular stencil, then so is $M_1$. Moreover, since $b_{max} = b_{i_1} = \ldots = b_{i_s}$, then the indices of the columns of $M_1$ are $(a_{i_1}, b_{max}), \ldots , (a_{i_s}, b_{max})$. We claim that the $s \times s$ sub-stencil $H_1'$ in $H_1$ formed by the rows $c_{i_1}, \ldots , c_{i_s}$ and the columns $a_{i_1}, \ldots , a_{i_s}$ is upper triangular, which will yield the inequality $\vrk{H_1} \ge s \ge k/n_2$. Because this holds for every full rank sub-stencil $M$ of $H_1 \tensor H_2$, then we conclude $\vrk{H_1} \ge \vrk{H_1 \otimes H_2}/n_2$. Thus it suffices for us to show that the sub-stencil $H_1'$ of $H_1$ is upper triangular. 

Since $M_1$ is upper triangular, we know that $M_1[(a_i, b), (c_i, d_i)] = \star$, and that implies $H_1[a_i, c_i] = H_2[b, d_i] = \star$. Moreover, for $i > j$, we know that $M_1[(a_j, b), (c_i, d_i)] = 0$, which means that one of the entries $H_1[a_j, c_i]$ or $H_2[b, d_i]$ is a zero, but since $H_2[b, d_i] = \star$, then we deduce that $H_1[a_j, c_i] = 0$. Thus we conclude that the $s \times s$ sub-stencil $H_1'$ in $H_1$ formed by the rows $c_{i_1}, \ldots , c_{i_s}$ and the columns $a_{i_1}, \ldots , a_{i_s}$ is an upper triangular stencil.
\end{proof}

\subsection{Visible rank of the tensor powers}

Naturally, one would be interested in tensoring a stencil $H$ with itself several times and then examine the properties arising from it.

\begin{defn}[Tensor power]
\label{def:tensor-power}
Given an $m \times n$ stencil $H$, the $k$'th tensor of $H$ is the $m^k \times n^k$ stencil $H^{\otimes k}$ defined as 
\begin{equation*}
    H^{\otimes k} \coloneqq \underbrace{H \otimes H \otimes \ldots \otimes H}_{\text{$k$ times}}
\end{equation*}
\end{defn}
By combining all the results from the previous subsections, we obtain the following corollary.
\begin{corollary}
\label{cor:tp-bounds}
For a natural number $k$  and an $m \times n$ stencil $H$, we have the inequality 
\begin{equation*}
\rk{\F}{H} \ge \rk{\F}{H^{\otimes k}}^{1/k} \ge \vrk{H^{\otimes k}}^{1/k} \ge \vrk{H}
\end{equation*}
Moreover, we also have the inequality $\vrk{H^{\otimes k}} \le n^{k-1}\vrk{H}$.
\end{corollary}

From the previous corollary, notice that by considering the visible rank of the higher tensor powers of $H$, one can attain better lower bounds on the $\F$-rank. As such, it seems natural to consider the best possible bound achieved through this vein.

\begin{defn}[Visible Capacity]
\label{def:stencil}
The visible capacity of a stencil $H$, denoted as $\Upsilon(H)$, is defined as $\Upsilon(H) \coloneqq \sup_k{\vrk{H^{\otimes k}}^{1/k}}$.
\end{defn}

By Corollary~\ref{cor:tp-bounds}, we deduce that $\rk{\F}{H} \ge \Upsilon(H)$ over any field $\F$. It is not known to us if there are stencils for which there is a gap between its visible capacity and all its $\F$-ranks. We leave the discussion of this point to Question~\ref{further:ark-vc-gap} in Section~\ref{further}.

\section{Tensor Powers of Stencils for $2$-DRGP Codes and $q$-LCCs}
\label{locality-tp}

In this section, we consider the tensor product from the previous section and use it to analyze the visible rank of the tensor powers of $2$-DRGP and $q$-LCC stencils in hopes of obtaining better lower bounds on the $\F$-rank via Corollary~\ref{cor:tp-bounds}.

\subsection{$2$-DRGP stencils}

In this subsection, we prove that the second tensor power of an arbitrary $2$-DRGP stencil has a large visible rank.

\begin{theorem}
\label{thm:drgp-tensor-lb}
For any $2$-DRGP stencil $H$, we have $\vrk{H \tensor H} \ge n$.
\end{theorem}
\begin{proof}
We cite~\cite{RV16, Woo16} for the proof of this part. We will follow the notations given in~\cite{Woo16} closely. While both proofs show that $\rk{\F}{H} \ge \sqrt{2n} - O(1)$, we will prove that $\rk{\F}{H} \ge \sqrt{n}$ by showing that $\vrk{H \otimes H} \ge n$ and then applying Corollary~\ref{cor:tp-bounds}. We rewrite their proofs in terms of stencils.

Consider the $n \times n$ sub-stencil $D$ in $H \otimes H$ whose columns are $(1,1), \ldots , (n,n)$ and whose rows are $((1, 1), (1,2)), \ldots , ((n,1), (n,2))$. We claim that $D$ has $\star$'s along the diagonal and zero everywhere else, which then implies that it is visibly full rank. Indeed, the entry $D[((i, 1), (i,2)), (j,j)]$ equals $\star$ if and only if both $H[(i,1), j]$ and $H[(i,2), j]$ are $\star$'s. Since $H$ is a 2-DRGP stencil, this happens precisely when $i=j$. Thus $D$ is a diagonal stencil.
\end{proof}

Thus by Corollary~\ref{cor:tp-bounds}, we obtain a lower bound $\rk{\F}{H} \ge \sqrt{n}$ for any field $\F$. On the other hand, the best known lower bounds yield $\rk{\F}{H} \ge \sqrt{2n} - O(1)$, and so one might be interested in achieving this lower bound through the viewpoint of tensor products. In order to improve the lower bound that we have, we first have to translate our proof into linear-algebraic terms.\smallskip

Given a field $\F$, suppose that we have an $\F$-witness $A$ of the $2$-DRGP stencil $H$ whose rank is $r$. Decompose $A = MN$ where $M$ is a $2n \times r$ matrix and $N$ is a $r \times n$ matrix. Denote the $i$'th column of $N$ by $w_i$. Then the proof of Theorem~\ref{thm:drgp-tensor-lb} is equivalent to saying that the tensors $\{w_i \tensor w_i\}_{i=1}^n$ are linearly independent. Since they live in the vector space $\F^r \otimes \F^r$,  we deduce the inequality $r^2 \ge n$, which gives us the same lower bound as we obtained in Theorem~\ref{thm:drgp-tensor-lb}. Now, if one is more careful about the vector space, one can notice that the tensors $\{w_i \tensor w_i\}_{i=1}^n$ are in fact contained in the space of \emph{symmetric} tensors, which has a dimension of $\binom{r+1}{2}$. Thus we obtain the inequality $\binom{r+1}{2} \ge n$, which gives us $r \ge \sqrt{2n} - O(1)$.

\subsection{$q$-LCC stencils}

In this subsection, we show that the visible rank of the $k$'th tensor power of a $q$-LCC stencil would not improve the current bound of $n - \widetilde{\Omega}(n^{(q-2)/(q-1)})$ in Theorem~\ref{thm:lcc-construction} for $k \le \mathrm{polylog}(n)$. More generally, we show that the visible rank of small tensor powers are not significantly bigger than the visible rank for the regime of high-rate stencils.

\begin{prop}
\label{prop:high-vrk-tp}
Let $H$ be an $m \times n$ stencil whose visible rank is at most $n - s$. For any fixed natural number $k$, we have $\vrk{H^{\otimes k}}^{1/k} \le n - \frac{s}{k}$.
\end{prop}
\begin{proof}
By Corollary~\ref{cor:tp-bounds} and the inequality $(1-x)^{1/k} \le 1 - \tfrac{x}{k}$ for $x \ge 0$, we have that
\begin{equation*}
    \vrk{H^{\otimes k}}^{1/k} \le \left(n^{k-1}\vrk{H}\right)^{1/k} \le \left(n^{k-1}(n-s)\right)^{1/k} = n\left(1- \frac{s}{n}\right)^{1/k} \le n\left(1- \frac{s}{kn}\right) = n - \frac{s}{k} \ . \qedhere
\end{equation*}
\end{proof}

From Proposition~\ref{prop:high-vrk-tp}, notice that the visible rank of a $n^{o(1)}$ tensor power of a $q$-LCC stencil would not improve the current bounds on $q$-LCCs by any polynomial factor. More formally, we deduce the following corollary.

\begin{corollary}
\label{thm:lcc-tensor-lb}
Let $H$ be a $q$-LCC stencil whose visible rank is at most $n - \widetilde{\Omega}(n^{(q-2)/(q-1)})$. For any natural number $k$, we have $\vrk{H^{\otimes k}}^{1/k} \le n - \widetilde{\Omega}(n^{(q-2)/(q-1)})/k$.
\end{corollary}

\section{Exponential Gaps between Tensor Powers}
\label{locality-tp}

In this section, we prove Theorem~\ref{thm:tensor-gaps}. Such exponential gaps have been shown in~\cite{AL06} for the Shannon capacities of graphs, and by following the same proof method given there, we are able to attain the same exponential gaps for the visible ranks of the tensor powers.\medskip

Before we prove the theorem, we will first define some notations that will help clarify the forthcoming analysis. For any vector $v = (v^{(1)}, \ldots , v^{(k)}) \in X^k$, let $S(v) = \{v^{(1)}, \ldots , v^{(k)}\}$ be the set of values of the coordinates of $v$. For $a \in [k]$, let $p_a(v) \in X^{k-1}$ be the vector $p_a(v) = (v^{(1)}, \ldots , v^{(a-1)}, v^{(a+1)}, \ldots , v^{(k)})$. Define the \emph{distinct rank} of $H^{\otimes k}$, denoted $\text{drk}(H)$, to be the largest visibly full rank substencil $H_d$ in $H^{\otimes k}$ such that $S(r_1) \cap S(r_2) = \varnothing$ for any distinct rows $r_1$ and $r_2$ in $H_d$ and $S(c_1) \cap S(c_2) = \varnothing$ for any distinct columns $c_1$ and $c_2$ in $H_d$. Moreover, we call such a substencil $H_d$ as \emph{distinctly full rank}. We first prove the following lemma, which will help reduce the analysis of Theorem~\ref{thm:tensor-gaps-2} to only look at the distinct visible ranks instead of the visible ranks.

\begin{lemma}
\label{lemma:distinct-vrk}
For any stencil $H$ and natural number $k \ge 2$, we have $\vrk{H^{\otimes k}} \le 2k^2\vrk{H^{\otimes (k-1)}}\text{drk}(H^{\otimes k})$.
\end{lemma}
\begin{proof}
Suppose we have an $\ell \times \ell$ substencil $H_0$ in $H^{\otimes k}$
whose general diagonal of stars is $S = \{(r_1, c_1), \ldots , (r_\ell, c_\ell)\}$. We are going to select a distinctly full rank $\ell_d \times \ell_d$ substencil $H_d$ from $H_0$ such that $\ell_d \ge \ell / 2k^2\vrk{H^{\otimes (k-1)}}$. Since $H_0$ was an arbitrary visibly full rank substencil, this will show our lemma.

Let's now construct $H_d$. Start with the set $S_d = \varnothing$ and repeat the following process: pick an arbitrary pair $(r_i, c_i)$ from $S$, add $(r_i, c_i)$ to $S_d$, and remove all pairs $(r_j,c_j)$ in $S$ for which we either have $S(r_i) \cap S(r_j) \neq \varnothing$ or $S(c_i) \cap S(c_j) \neq \varnothing$. Once this process ends, we will have $S(r_i) \cap S(r_j) = \varnothing$ and $S(c_i) \cap S(c_j) = \varnothing$ for any distinct pairs $(r_i, c_i)$ and $(r_j, c_j)$ in $S_d$. Thus the substencil $H_d$ formed by the pairs in $S_d$ will be distinctly full rank.

It remains to show that $\abs{S_d} \ge \abs{S} / (2k^2\vrk{H^{\otimes (k-1)}})$, and we show this by proving that every step of the process above only deletes at most $2k^2\vrk{H^{\otimes (k-1)}}$ pairs. Indeed, consider any pair $(r_i, c_i)$. There are at most $k$ values in $S(r_i)$, and each of them has $k$ possible coordinates to be in any other row $r_j$. Thus there are at most $k^2$ possibilities for any value in $S(r_i)$ to be in any coordinate. Now, for any $x \in S(r_i)$ and $a \in [k]$, consider the diagonal $S_{a,x} = \{(r, c) \in S \; | \; r^{(a)} = x\}$. We claim that $\abs{S_{a,x}} \le \vrk{H^{\otimes (k-1)}}$. Once we prove this claim, then it would imply that the constraint $S(r_i) \cap S(r_j) = \varnothing$ will delete at most $k^2\vrk{H^{\otimes (k-1)}}$ pairs. By the same argument for $c_i$, we would also deduce that we delete at most $k^2\vrk{H^{\otimes (k-1)}}$ pairs to obtain the constraint $S(c_i) \cap S(c_j) = \varnothing$, totaling to at most $2k^2\vrk{H^{\otimes (k-1)}}$ deletions. Thus it remains to show that $\abs{S_{a,x}} \le \vrk{H^{\otimes (k-1)}}$.

Indeed, since $H_d$ is a visibly full rank substencil, then so is the substencil $H_d^{a,x}$ formed by the pairs in $S_{a,x}$. Now, because $S_{a,x}$ forms a star diagonal in $H^{\otimes k}$, then we know that $H^{\otimes k}[r,c] = \star$ for every $(r,c) \in S_{a,x}$. From the definition of $H^{\otimes k}$ and $S_{a,x}$, we deduce that $H[r^{(a)}, c^{(a)}] = H[x,c^{(a)}] = \star$ for all $(r,c) \in S_{a,x}$. This then implies that $H[r_1^{(a)},c_2^{(a)}] = \star$ for every row $r_1$ and column $c_2$ appearing in $S_{a,x}$. Since $H_d^{a,x}$ is visibly full rank, then this means that the substencil in $H^{\otimes (k-1)}$ formed by the pairs $\{(p_a(r), p_a(c)) \; | \; (r,c) \in S_{a,x}\}$ is also visibly full rank as the $a$'th value in any product in $H_d^{a,x}$ is always a $\star$.
Moreover, the map $p_a$ is injective over the rows in $S_{a,x}$ as for any distinct rows $r_1$ and $r_2$ appearing in $S_{a,x}$, we know that $r_1^{(a)} = r_2^{(a)} = x$ but $r_i \neq r_j$. Thus the substencil in $H^{\otimes (k-1)}$ formed by the pairs $\{(p_a(r), p_a(c)) \; | \; (r,c) \in S_{a,x}\}$ has the same size as $H_d^{a,x}$. Since it is visibly full rank, then we conclude that $\abs{S_{a,x}} \le \vrk{H^{\otimes (k-1)}}$ for any $a \in [k]$ and $x \in S(r_i)$, and this finishes the proof.
\end{proof}

Now, we proceed to prove our main theorem for this section.

\begin{theorem}
\label{thm:tensor-gaps-2}
For any fixed natural number $t \ge 2$ and sufficiently large $n$, there exists a $tn \times n$ stencil $H$ such that $\vrk{H^{\otimes k}} = O_t((\log{n})^k)$ for any $k = 1, \ldots , t-1$ and $\vrk{H^{\otimes t}} \ge n$.
\end{theorem}
\begin{proof}
We proceed to prove this theorem by showing that a random stencil with a specified structure satisfies the properties with high probability. Now, index the columns by $[n]$ and the rows by $[n] \times [t]$. Define the collection of entries $S_{i,i'} \coloneqq \{H[(i,1),i'], \ldots , H[(i,t),i']\}$. Set $H[(i,j),i] = \star$ for $(i,j) \in [n] \times [t]$. Now, for $i' \neq i \in [t]$, uniformly set one of the $t$ entries in $S_{i,i'}$ to zero while the remaining $t-1$ entries are set to $\star$. We are going to show that such a random stencil $H$ satisfies the bounds we outlined above with high probability.

First, let us show that $\vrk{H^{\otimes t}} \ge n$ for any such stencil $H$. Indeed, consider the $n \times n$ substencil $H'$ in $H^{\otimes t}$ whose rows are $\{((i,1), \ldots , (i,t)) \; | \; i \in [n]\}$ and whose columns are $\{(i, \ldots , i) \in [n]^t\}$. For any row $r_i = ((i,1), \ldots , (i,t))$ and column $c_{i'} = (i', \ldots , i')$, we know that $H^{\otimes t}[r_i, c_{i'}] = \star$ if and only if $H[(i,1),i'] = \ldots = H[(i,t),i'] = \star$. By definition of $H$, that is only possible if and only if $i = i'$. Thus $H'$ is a star diagonal, meaning that it is visibly full rank. Thus $\vrk{H^{\otimes t}} \ge n$.

Now consider the case when $k < t$. We will show that $\text{drk}(H^{\otimes k}) < 16k^2t^2\log{n}$ with high probability for all $k = 1, \ldots , t-1$. Our theorem will then follow by inductively applying Lemma~\ref{lemma:distinct-vrk} on $\vrk{H^{\otimes k}}$. Indeed, fix a positive integer $M_d$ and a $M_d \times M_d$ distinctly full rank substencil $H_d$ in $H^{\otimes k}$. By Lemma~\ref{lemma:upper-tri}, we know that we can reorder the rows and columns of $H_d$ so that $H_d[r_i, c_i] = \star$ and $H_d[r_i, c_{i'}] = 0$ for $i > i'$. For $i > i'$, let $E_{i,i'}$ be the event that $H_d[r_i, c_{i'}] = 0$. By applying a Union Bound, we have that
\begin{equation} \label{eq:event-bound}
    \Pr{E_{i,i'}} \le \sum_{a=1}^k{\Pr{H[r_i^{(a)}, c_{i'}^{(a)}] = 0}} \le \sum_{i=1}^k{\frac{1}{t}} = \frac{k}{t}
\end{equation}
Moreover, notice that the event $H[r_i^{(a)}, c_{i'}^{(a)}] = 0$ depends on at most $t-1$ other such events as any zero entry in $H$ is chosen from a collection of $t$ entries. Since $E_{i,i'}$ is a conjunction of the events $H[r_i^{(a)}, c_{i'}^{(a)}] = 0$ for $a \in [k]$, then $E_{i,i'}$ depends on at most $k(t-1)$ other events of the form $H[r,c] = 0$. Since $H_d$ is distinctly full rank, then each of those $k(t-1)$ other events can only appear at most once among all the events $\{E_{i,i'}\}_{i > i'}$. Thus by a simple greedy algorithm, we can find a subcollection of events $T \subseteq \{E_{i,i'}\}_{i > i'}$ with $\abs{T} \ge \frac{\binom{M_d-1}{2}}{k(t-1)+1} \ge \frac{(M_d-2)^2}{2kt}$
such that all the events in $T$ are mutually independent of each other. Thus by the independence of the events $T$, Inequality~\eqref{eq:event-bound}, and the inequality $k < t$, we find that
\begin{equation*}
    \Pr{\text{$H_d$ is visibly full rank}} \le \Pr{\land_{E \in T}{E}} = \prod_{E \in T}{\Pr{E}} \le \left(\frac{k}{t}\right)^{\abs{T}} \le \left(1 - \frac{1}{t}\right)^{\frac{(M_d-2)^2}{2kt}} \le e^{-\frac{(M_d-2)^2}{2kt^2}}
\end{equation*}
And so by applying a Union Bound over all possible substencils $H_d$ and its orderings of its rows and columns, we find that
\begin{equation*}
    \Pr{\text{dvrk}(H^{\otimes k}) \ge M_d} \le (tn)^{kM_d}n^{kM_d}e^{-\frac{(M_d-2)^2}{2kt^2}} = \left(t^{\frac{k}{\log{n}}}e^{2k - \frac{M_d - 4 + \frac{4}{M_d}}{2kt^2\log{n}}}\right)^{M_d\log{n}}
\end{equation*}
And so if we pick $M_d = 16k^2t^2\log{n}$ and large enough $n$, then we get that $\text{dvrk}(H^{\otimes k}) < 16k^2t^2\log{n}$ with high probability for $k = 1, \ldots , t-1$. From the definition of $\text{drk}(\cdot)$, we see that $\vrk{H} = \text{drk}(H)$. Thus by inductively applying Lemma~\ref{lemma:distinct-vrk} for $k < t$, we deduce that the inequality
\begin{equation*}
    \vrk{H^{\otimes k}} \le 2^{k-1}(k!)^2\prod_{i=1}^k{\text{drk}(H^{\otimes i})} < 2^{k-1}(k!)^2\prod_{i=1}^k{[16i^2t^2\log{n}]} = 2^{5k-1}(k!)^4t^{2k}(\log{n})^k
\end{equation*}
holds with high probability. Thus such a construction $H$ exists.
\end{proof}

\section{Further Directions and Discussion}
\label{further}

Stencils provide an initial framework toward combinatorial methods for effectively lower bounding the rank of a matrix. However, we have seen the limitations of the visible rank with $2$-DRGP stencils as well as small tensor powers of $q$-LCC stencils. We leave the reader with questions that remain open about the current framework and possibilities of imposing further restrictions on the model to obtain sharper lower bounds on the rank for restricted cases.

\begin{enumerate}

    \item While we may have shown that the $k$'th tensor power of a $q$-LCC does not yield better lower bounds for $k \le n^{o(1)}$, this does not rule out the possibility that the visible capacity might yield better lower bounds. In fact, we do not know if there are any stencils for which the visible capacity does not match the lowest possible rank for the stencil. In other words, does there exist a stencil $H$ such that $\rk{\F}{H} > \Upsilon(H)$ for every field $\F$?.\label{further:ark-vc-gap}

    \item In this paper, we shown a polynomial gap between $\rk{\F}{H}$ and $\vrk{H}$ by proving that there are $2$-DRGP stencils $H$ with $\vrk{H} = O(\log{n})$ and $\rk{\F}{H} = \Omega(\sqrt{n})$. On the other hand, can there also be a similar polynomial gap with the quantities $n - \rk{\F}{H}$ and $n - \vrk{H}$? From Proposition~\ref{prop:high-vrk-tp}, we have seen that the visible ranks of the $k$'th tensor power for $k \le n^{o(1)}$ would not suffice to show this polynomial gap. Nonetheless, it still leaves the possibility of using the visible capacity to show this polynomial gap, but we do not know of any methods that can lower bound the visible capacity other than the visible ranks of finite tensor powers. Note that this question is the symmetric spanoid version of Question 2 posed in~\cite{DGGW20}.\label{further:co-gap}

    \item One natural restriction is to consider attaining lower bounds for $\rk{\F}{H}$ over all fields $\F$ such that $\abs{\F} \le M$ for some $M \in \N$. In such cases, methods such as those used in Proposition~\ref{prop:zero-set-ub} would not work. More generally, suppose that we are interested in attaining lower bounds for \emph{$k$-color stencils}, a stencil wherein any algebraic witness of it can only use at most $k$ distinct nonzero entries. Can one consider better combinatorial proxies for the algebraic rank of $k$-color stencils (particularly constant-color stencils)?\smallskip
    
    It is possible to attain asymptotically better lower bounds on the algebraic rank for colored stencils. Indeed, consider a $k$-color stencil $H \in \{0, \star\}^{n \times n}$ where $H[i, i] = 0$ and $H[i, j] = \star$ for $i \neq j$. We know that $\vrk{H} = 2$, but since $H$ is a $k$-color stencil, then one can show by a polynomial argument (which we leave for the reader to try) that any algebraic witness of $H$ has a rank of at least $\Omega_k(n^{1/k})$. \label{further:color-stencils}

    \item One component that stencils do not capture is the distance of a code, but that's partly due to the distance of a code being dependent on the field size. If we consider $k$-color stencils, can one find combinatorial proxies that upper and lower bound the distance of any algebraic witness of the stencil? \label{further:dist-realized}

    \item Can one explicitly construct an $n \times n$ stencil $H$ whose visible rank is robust to small entry corruptions? Any algebraic witness of $H$ would automatically be a rigid matrix. \label{further:robust-vrk}

\end{enumerate}

Finally, our aim in introducing the visible rank is to study combinatorial proxies for the rank of a matrix. Are there other combinatorial quantities associated with a matrix that could yield more effective lower bounds on its rank?

\bibliographystyle{alpha}
\bibliography{vrk}

\end{document}